\documentclass{llncs}
\usepackage[utf8]{inputenc}
\usepackage[T1]{fontenc} 

\usepackage{graphicx} % obrazki
\usepackage{amssymb} % dodatkowe znaczki matematyczne
\usepackage{mathtools}
\usepackage{algorithmic}
\usepackage{authblk}
\usepackage{color}
\usepackage{tikz}

\newcommand{\va}[1]{{\color{black}\sl  #1}}
\begin{document}

\title{Improved approximation algorithm for $k$-level UFL with penalties, a simplistic view on randomizing the scaling parameter}

\author{Jaroslaw Byrka\inst{1} \thanks{jby@ii.uni.wroc.pl, supported by FNP HOMING PLUS/2010-1/3 grant and MNiSW grant number N N206 368839, 2010-2013 } \and Shanfei Li\inst{2} \thanks{shanfei.li@tudelft.nl}\and Bartosz Rybicki\inst{1} \thanks{bry@ii.uni.wroc.pl, research supported by NCN 2012/07/N/ST6/03068 grant}}

\institute{Institute of Computer Science, University of Wroclaw, Poland,\\
\and
Delft Institute of Applied Mathematics, TU Delft, The Netherlands}

\maketitle

\begin{abstract}
  The state of the art in approximation algorithms for facility location problems are complicated combinations of various techniques.
  In particular, the currently best 1.488-approximation algorithm for the uncapacitated facility location (UFL) problem by Shi Li is presented as a result of a non-trivial randomization of a certain scaling parameter
  in the LP-rounding algorithm by Chudak and Shmoys combined with a primal-dual algorithm of Jain et al.
  In this paper we first give a simple interpretation of this randomization process in terms of solving an auxiliary (factor revealing) LP.
  \va{Then, armed with this simple view point, we exercise the randomization on a more complicated algorithm for the $k$-level version of the problem with penalties in which the planner has the option to pay a penalty instead of connecting chosen clients, which results in an improved approximation algorithm.}
\end{abstract}

\section{Introduction}

In the uncapacitated facility location (UFL) problem the goal is to open facilities in a subset of given locations and connect each client to an open facility so as to minimize the sum of opening costs and connection costs. In the penalty avoiding (prize collecting) variant of the problem, a fixed penalty can be paid instead of connecting a client.

In the $k$-level uncapacitated facility location problem with penalties ($k$-level UFLWP), we are given a set $C$ of clients and a set $F = \bigcup_{t = 1}^{k} F_{l_t}$ of facilities (locations to potentially open a facility) in a metric space. Facilities are of $k$ different types (levels), e.g., for $k=3$ one may think of these facilities as shops, warehouses and factories.
Each set $F_{l_t}$ contains all facilities on level $t$ and the sets $F_{l_t}$ are pairwise disjoint.
Each client $j$ can either be connected to precisely one facility at each of $k$ levels (via a path), or be rejected in which case the penalty $p_j$ must be paid ($p_j$ can be considered as the loss of profit).
To be more precise, for a client $j$ to be connected, it must be connected with a path $(j,i_{1}, i_{2}, \cdots, i_{k-1}, i_{k})$, where $i_t$ is an open facility on level $t$.
The cost of connecting points $i,j \in C \cup F$, is the distance between $i$ and $j$, denoted by $c_{ij}$. The cost of opening facility $i$ is $f_{i}$ ($f_{i}\geq 0$).
The goal is to minimize the sum of the total cost of opening facilities (at all levels), the total connection cost and the total penalty cost.
In the uniform version of the problem all penalties are the same, i.e., for any two clients $j_1,j_2 \in C$ we have $p_{j_1} = p_{j_2}$.

\subsection{Related work}

If $p_j, j\in C$ are big enough, $k$-level UFLWP is the $k$-level UFL problem, for which Krishnaswamy and Sviridenko \cite{Krishnaswamy} showed
$1.61$-hardness of approximation for general $k$ and $1.539$-hardness for $k=2$. Actually, even for $k=1$ Guha and Khuller \cite{Guha} showed that the approximation ratio is at least $1.463$, unless $NP\subseteq DTIME(n^{\log\log n})$. The current best known approximation ratio for this simplest case $k=1$ is $1.488$ by Li \cite{ShiLi}.

For $2$-level UFL problem Shmoys, Tardos, and Aardal \cite{shmoys_stoc97} gave the first constant factor approximation algorithm by extending the algorithm for $1$-level and obtaining an approximation ratio $3.16$. Subsequently, Aardal, Chudak, and Shmoys \cite{Aardal} used randomized rounding to get the first algorithm for general $k$, which had approximation ratio of $3$.
Ageev, Ye and Zhang \cite{Ageev} gave a combinatorial $3.27$-approximation algorithm for general $k$ by reducing the $k$-level directly into $1$-level problem. By recursive reduction, i.e., reducing $k$-level to $k-1$ level, they obtained an improved $2.43$-approximation for $k=2$ and $2.85$ for $k=3$. Later, this was improved by Zhang \cite{Zhang}, who combined the maximization version of $1$-level UFL problem and dual-fitting to get a $1.77$-approximation algorithm for $k=2$, and a $2.53$-approximation for $k=3$.
Byrka and Aardal \cite{Byrka} improved the ratio for $k=3$ to $2.492$.
For $k>2$ the ratio was recently improved by Byrka and Rybicki \cite{Rybicki} to $2.02$ for $k=3$, $2.14$ for $k=4$, and the ratio converges to 3 when $k\rightarrow + \infty$.

UFL with penalties was first introduced by Charikar et al. \cite{Charikar}, who gave a $3$-approximation algorithm based on a primal-dual method. Later, Jain et al. \cite{Jain} indicated that their greedy algorithm for UFL could be adapted to UFLWP with the approximation ratio $2$. Xu and Xu \cite{Xu:2005,Xu:2009} proposed a $2.736$-approximation algorithm based on LP-rounding and a combinatorial $1.853$-approximation algorithm by combining local search with primal-dual. Later, Geunes et al. \cite{Geunes} presented an algorithmic framework which can extend any LP-based $\alpha$-approximation algorithm for UFL to get an $(1-e^{-1/\alpha})^{-1}$-approximation algorithm for UFL with penalties. As a result, they gave a $2.056$-approximation algorithm for this problem. Recently, Li et al. \cite{YuLi} extended the LP-rounding algorithm by Byrka and Aardal \cite{Byrka} and the analysis by Li \cite{ShiLi} to UFLWP to give the currently best $1.5148$-approximation algorithm.

For multi-level UFLWP, Asadi et al. \cite{Asadi} presented an LP-rounding based $4$-approximation algorithm by converting the LP-based algorithm for UFLWP by Xu and Xu \cite{Xu:2005} to $k$-level. To the best of our knowledge, this is the only algorithm for multi-level UFLWP in the literature.

\subsection{Our contribution}

We first show that algorithms whose performance can be analysed with a linear function of certain instance parameters, like the Chudak and Shmoys algorithm \cite{Chudak} for UFL, can be easily combined and analysed with a natural factor revealing LP. This simplifies the argument of Shi Li~\cite{ShiLi} for his $1.488$-approximation algorithm for UFL \va{since an explicit distribution for the parameters obtained by a linear program is not necessary in our factor revealing LP}.

With this tool one can easily randomize the scaling factor in LP-rounding algorithms for various variants of the UFL problem. \va{We demonstrate this by randomizing the algorithm for $k$-level UFLWP. For $k$-level UFL we can get the same approximation ratios as for $k$-level UFLWP by setting $p_j=+\infty, j\in C$.}

Note that the previously best ratio is 4 for $k$-level UFLWP ($k\geq 2$) \cite{Asadi} and $1.5148$ for $k=1$  \cite{YuLi}. The following table shows how much we improve the approximation ratios of our algorithm for $k = 1, \dots, 10$ by involving randomization of the scaling factor. \va{Irrespective of the way in which we choose $\gamma$, deterministically or randomly, approximation ratio converges to three.}
%\footnote{Note that for $k$-level UFL problem we currently do not have a (1.11, 1.78)-approximation algorithm. Thus, although we introduce randomization of the scaling factor to $k$-level UFL problem (with penalties), the approximation ratio obtained by our algorithm is greater than 1.488 for $k=1$. An algorithm for UFL is called (a,b)-approximation if the cost of returned solution is upper bounded by $a \cdot F^* + b \cdot C^*$, where $F^*$ and $C^*$ are, respectively, the costs of establishing connections and opening facilities in an optimal solution.}. 
\begin{table}
\centering
  \begin{tabular}{ c | c | c | c | c | c | c | c | c | c | c }
    $k$ & 1 & 2 & 3 & 4 & 5 & 6 & 7 & 8 & 9 & 10 \\ \hline
    no randomization of $\gamma$ & 1.58 & 1.85 & 2.02 & 2.14 & 2.24 & 2.31 & 2.37 & 2.42 & 2.46 & 2.50 \\ \hline
    with randomization of $\gamma$ & 1.52 & 1.79 & 1.97 & 2.09 & 2.19 & 2.27 &2.33 & 2.39 & 2.43 & 2.47 \\
  \end{tabular}
  \caption{Comparison of ratios.}
  \label{improved_ratios}
\end{table}

\section{Simple version of Li's argument}

Consider the following standard LP relaxation of UFL.

\begin{eqnarray}
\label{lp_ufl:goal}
  min \sum_{i \in F}{\sum_{j \in C} {c_{ij}x_{ij}}} &+& \sum_{i \in F} y_i f_i\\
 \label{lp_ufl:connected}
  \sum_{i \in F} x_{ij}&=& 1 ~~~~~\forall_{j \in C}\\
   y_i - x_{ij}&\geq& 0 ~~~~~\forall_{i \in F, j \in C} \\
   x_{ij}, y_{i} &\geq& 0 ~~~~~\forall_{i \in F, j \in C}
\end{eqnarray}

Chudak and Shmoys \cite{Chudak} gave a randomized rounding algorithm for UFL based on this relaxation. Later
Byrka and Aardal \cite{Byrka} considered a variant of this algorithm where the facility opening variables were initially scaled up by a factor of $\gamma$. They showed that
for $\gamma \geq \gamma_0 \approx 1.67$ the algorithm returns a solution with cost at most $\gamma$ times the fractional facility opening cost plus $1+2e^{-\gamma}$ times the fractional connection cost. 
This algorithm, when combined with the (1.11, 1.78)-approximation algorithm of Jain, Mahdian and Saberi~\cite{Jain} (JMS algorithm for short), 
is easily a 1.5-approximation algorithm for UFL. More recently, Li \cite{ShiLi} showed that 
\va{by randomly choosing the scaling parameter $\gamma$ from an certain probability distribution one obtains an improved 1.488-approximation algorithm. 
A natural question is what improvement this technique gives in the $k$-level variant.}

In what follows we present our simple interpretation and sketch the analysis of the randomization by Li. We argue that a certain factor revealing LP provides a valid upper bound on the obtained approximation ratio. The appropriate probability distribution for the scaling parameter (engineered and discussed in detail in \cite{ShiLi}) may in fact be directly read from the dual of our LP. While we do not claim to get any deeper understanding of the randomization process itself, the simpler formalism we propose is important for us to apply randomization to a more complicated algorithm for $k$-level UFL, which we describe next.

\subsection{Notation}

Let $F_j$ denote the set of facilities which client $j \in C$ is fractionally connected to, i.e., facilities $i$ with $x_{ij} > 0$ in the optimal LP solution. Since for uncapacitated facility location problems one can split facilities before rounding, to simplify the presentation, we will assume that $F_j$ contains lots of facilities with very small fractional opening $y_i$. This will enable splitting $F_j$ into subsets of desired total fractional opening.

\begin{definition}[definition 15 from \cite{ShiLi}]
 Given an UFL instance and its optimal fractional solution $(x^*, y^*)$, the characteristic function $h_j : [0, 1] \longmapsto R$ of a client $j \in C$ is the following. Let $i_1, i_2, \cdots, i_m$ denote the facilities in $F_j$, in a non-decreasing order of distances to j. Then $h_j(p) = d(i_t, j)$, where $t$ is the minimum number such that $\sum_{s=1}^{t}y^*_{i_s} \geq p$. Furthermore, define $h(p) = \sum_{j \in C} h_j(p)$ as the characteristic function for the entire fractional solution.
\end{definition}

\begin{definition}
 Volume of set $F' \subseteq F$, denoted by $vol(F')$ is the sum of facility openings in that set, i.e., $vol(F') = \sum_{i \in F'} y^*_i$.
\end{definition}

 For $l = 1, 2 \dots, n$ define $\gamma_l = 1 + 2 \cdot \frac{n - l}{n}$, which will form the support for the probability distribution of the scaling parameter $\gamma$. Suppose that all facilities are sorted in an order of non-decreasing distances from client $j \in C$. Scale up all $y^*$ variables by $\gamma_l$ and divide the set of facilities $F_j$ into two disjoint subsets: the close facilities of client $j$, $F_j^{C_l}$, such that $vol(F_j^{C_l}) = 1$; and the distant facilities $F_j^{D_l} = F_j \setminus F_j^{C_l}$. Note that $vol(F_j^{D_l}) = \gamma_l - 1$.
Observe that$\frac{1}{\gamma_k} < \frac{1}{\gamma_l} \Rightarrow F_j^{C_k} \subset F_j^{C_l} \wedge F_j^{C_l} \setminus F_j^{C_k} \neq \emptyset$. We now split $F_j$ into disjoint subsets $F_j^l$. Define $F_j^{C_0} = \emptyset$ and $F_j^l = F_j^{C_l} \setminus F_j^{C_{l-1}}$, where $l = 1, 2 \dots, n$. The average distance from $j$ to facilities in $F_j^l$ is $c_l(j) = \int_{1/\gamma_{l-1}}^{1/\gamma_{l}} h_j(p)~dp$ for $l > 1$ and $\int_{0}^{1/\gamma_{1}} h_j(p)~dp$ for $l = 1$. Note that $c_{l}(j) \leq c_{l+1}(j)$ and $D_{max}^{l}(j) \leq c_{l+1}(j)$, where $D_{max}^{l}(j) = max_{i \in F_j^l}{c_{ij}}$.

Since the studied algorithm with the scaling parameter $\gamma = \gamma_k$ opens each facility $i$ with probability $\gamma_k \cdot y_i^*$,
and there is no positive correlation between facility opening in different locations,
the probability that at least one facility is open from the set $F_j^l$ is at least $1 - e^{-\gamma_k \cdot vol(F_j^l)}$.

Crucial to the analysis is the length of a connection via the cluster center $j'$ for client $j$ when no facility in $F_j$ is open.
Consider the algorithm with a fixed scaling factor $\gamma=\gamma_k$, an arbitrary client $j$ and its cluster center $j'$.
Li gave the following upper bound on the expected distance from $j$ to an open facility around its cluster center $j'$.
 \begin{lemma} [Lemma 14 from \cite{ShiLi}]
 \label{cluster_close_distance}
  If no facility in $F_j$ is opened, the expected distance to the open facility around $j'$ is at most $\gamma_k D_{av}(j) + (3 - \gamma_k)D_{max}^k(j)$,
  where $ D_{av}(j) = \sum_{i \in F_j} c_{ij}x_{ij}^*$.
 \end{lemma}

 \begin{corollary}
 \label{upperbound_for_c}
  If $\gamma=\gamma_k$, then the expected connection cost of client $j$ is at most
  $$E[C_j] \leq \sum_{l = 1}^{n} c_l(j) \cdot p_l + (1 - e^{-\gamma_k}) \cdot (\gamma_k D_{av}(j) + (3 - \gamma_k)D_{max}^k (j))$$ where $p_l$ is the probability of the following event: no facility is opened in distance at most $D_{max}^{l-1}(j)$ and at least one facility is opened in $F_j^l$.
 \end{corollary}

 \subsection{Factor revealing LP}
 Consider running once the JMS algorithm and the Chudak and Shmoys algorithm multiple times, one for each choice of the value for the scaling parameter $\gamma = \gamma_l = 1 + 2 \cdot \frac{n - l}{n}, l=1, 2 \dots n$. Observe that the following LP captures the expected approximation factor of the best among the obtained solutions, where $p_1^k = 1 - e^{-\frac{\gamma_k}{\gamma_1}}$  and $p_l^k = e^{-\frac{\gamma_k}{\gamma_{l-1}}} - e^{-\frac{\gamma_k}{\gamma_l}}$ for all $l > 1$. Goal of the below LP is to construct the worst case instance of distances $c_l$.
 \begin{eqnarray}
  \label{total_lp:max}
  max~T && \\
  \label{total_lp:t_upperbound}
  \gamma_k f + \sum_{l = 1}^{n} c_l \cdot p_l^k + (1 - e^{-\gamma_k}) (\gamma_k c + (3 - \gamma_k)c_{l+1}) \geq T &&  ~\forall_{k < n}\\
  \label{total_lp:jms}
  1.11 f + 1.78 c \geq T && \\
  \label{total_lp:c_constraint}
  \frac{1}{\gamma_1} \cdot c_1 + \sum_{i = 2}^{n}(\frac{1}{\gamma_i} - \frac{1}{\gamma_{i-1}}) \cdot c_i = c && \\
  \label{total_lp:c_order}
  0 \leq c_i \leq c_{i+1} \leq 1&&~\forall_{i < n}\\
  \label{total_lp:opt_sol}
  f + c = 1 && \\
  f, c\geq 0 &&
\end{eqnarray}
 
The variables of this program encode certain measurements of the function $h(p)$ defined for an optimal fractional solution. Intuitively, these are average distances
between a client and a group of facilities, summed up for all the clients. 
The program models the freedom of the adversary in selecting cost profile $h(p)$ to maximize the cost of the best of the considered algorithms. 
Variables $f$ and $c$ model the facility opening and client connection cost in the fractional solution. Inequality (\ref{total_lp:t_upperbound}) correspond to 
LP-rounding algorithms with different choices of the scaling parameter $\gamma$. Note that $D_{av}(j) = c$ and $D_{max}^l \leq c_{l+1}(j)$ holds for each client, that fact, with corollary (\ref{upperbound_for_c}), justifies inequality (\ref{total_lp:t_upperbound}). Inequality (\ref{total_lp:jms}) corresponds to the JMS algorithm~\cite{Jain}, 
and equality (\ref{total_lp:c_constraint}) encodes the total connection cost
%\va{, distance to set of facilities times volume of that set}
.

Interestingly, the choice of the best algorithm here is not better in expectation than a certain random choice between the algorithms. To see this, consider the dual of the above LP. In the dual, the variables corresponding to the primal constraints (\ref{total_lp:t_upperbound}) and (\ref{total_lp:jms}) simply encode the probabilities for choosing a particular algorithm. Our computational experiments with the above LP confirmed the correctness of the analysis of Li~\cite{ShiLi}. Additionally, from the primal program with distances we obtained the worst case profile $h(p)$ for the state of the art collection of algorithms considered (see Fig.~\ref{fig:hardcase} and Fig.~\ref{fig:probabilities} respectively for a plot of this tight profile and the distributions of the scaling factor for $k$-level UFL on different number of levels).

\section{Reduction from $k$-level UFL with uniform penalties to $k$-level UFL}
\va{The difficulty of $k$-level UFLWP lies in the extra choice of each client, that is, the penalty. 
We will explain how to overcome the penalties by converting the instance of UFLWP to an appropriate instance of UFL. We first consider the easy case of uniform penalties.}
\begin{lemma}
  \label{uniform-lemma}
  Each instance of UFL with uniform penalties can be modified to {an} appropriate UFL instance.
\end{lemma}

\begin{proof}
We can treat the penalty of client $j \in C$ as a facility at distance $p_j$ to client $j$ with opening cost zero. The distance from client $j$ to the penalty-facility of client $j'$ is equal to $c_{j, j'} + p_{j'}$. Note that $p_{j'} = p_{j}$. We can run any algorithm for UFL on the modified instance as described above. If in the obtained solution client $j$ is connected with the penalty-facility of client $j'$, we can switch $j$ to {its} penalty-facility without increasing the cost of the solution.
\qed
\end{proof}

Lemma \ref{uniform-lemma} implies that for $k$-level uncapacitated facility location with uniform penalties we have the following approximation ratios. Algorithms for $k = 1$ and $2$ are described in \cite{ShiLi} and \cite{Zhang}, for $k>2$ are described in this article.

\begin{center}
  \begin{tabular}{ c | c | c | c | c | c | c | c | c | c | c }
    $k$ & 1 & 2 & 3 & 4 & 5 & 6 & 7 & 8 & 9 & 10 \\ \hline
    ratio & 1.488 & 1.77 & 1.97 & 2.09 & 2.19 & 2.27 &2.33 & 2.39 & 2.43 & 2.47 \\
  \end{tabular}
\end{center}

Note that the reduction above does not work for the non-uniform case, because then the distance from client $j$ to the penalty-facility of client $j'$ could be smaller than~$p_j$.
Nevertheless we will show that LP-rounding algorithms in this paper can be easily extended to the non-uniform penalty variant.

\section{Extended LP formulation}

\va{For non-uniform case,} our algorithm is based on rounding a solution to the extended LP-relaxation of the problem. This extended LP may either be seen as the standard LP on a modified graph (see Appendix~\ref{sec:modiffication})
as described in~\cite{Rybicki}, or originate from the $k$-th level of the Sherali Adams hierarchy, or explicitly be written in terms of paths on the original instance. Here we use the explicit construction.
Note that in the optimal solution to $k$-level UFLWP each facility is connected to at most one facility on the higher level. We will impose this structure on the fractional solution
by creating multiple copies of the original facility, one for each path across the higher levels of facilities.

To describe the linear program we have to give a few definitions. Let $P_C$ be the set of paths which start in a client and end in a facility on level $k$. Let $P_t$ be the set of paths which start on level $t$ and end on the highest level $k$, i.e., in a root of some tree.
By $P$ we denote the set of all paths, i.e.,  $P = P_C \cup \bigcup_{t=1}^{k} P_{t}$. The cost of the path denoted by $c_p$ depends on the kind of path.
If $p = (j, i_1, i_2, \cdots, i_k) \in P_C$, then $c_p = c_{i_1 j} + c_{i_2 i_1} + \cdots + c_{i_{k}, i_{k-1}}$.
If $p = (i_t, i_{t+1}, \cdots, i_k) \in P_{t}$, then $c_p = f_{i_t}$.

\begin{eqnarray}
  \label{lp:min}
  min \sum_{p \in P} x_{p}c_{p} &+& \sum_{j \in C} g_j p_j\\
 \label{lp:out_one}
  \sum_{p \in P_{C} : j \in p} x_{p} + g_j&\geq& 1 ~~~~~\forall_{j \in C}\\
  \label{lp:thomas_order}
  x_{(i_{t+1}, i_{t+2}, \ldots i_{k})} -x_{(i_{t}, i_{t+1}, \ldots i_{k})} &\geq& 0 ~~~~~\forall_{p = (i_{t}, i_{t+1}, \ldots i_{k}) \in P_t, t<k}\\
  \label{lp:f_open_enough}
  x_{q} - \sum_{p = (j, \ldots i_{t}, i_{t+1} \ldots i_{k}) \in P_{C}} x_{p} &\geq& 0 ~~~~~\forall_{j \in C}, \forall_{q = (i_{t}, i_{t+1}, \ldots i_{k}) \in P \setminus P_{C}}\\
  x_{p} &\geq& 0 ~~~~~\forall_{p \in P} \\
  g_{j} &\geq& 0 ~~~~~\forall_{j \in C}
\end{eqnarray}

The natural interpretation of the above LP is as follows.
Inequality (\ref{lp:out_one}) states that each client is assigned to at least one path or is rejected.
Inequality (\ref{lp:thomas_order}) encodes that opening a lower level facility implies opening its unique higher level facility.
The most complicated inequality (\ref{lp:f_open_enough}) for a client $j \in C$ and a facility $i_{t} \in F_{l_t}$,
imposes that the opening of $i_t$ must be at least the total usage of it by the client $j$. Let $(x^*, g^*)$ be an optimal solution to the above LP.

\section{Algorithm for $k$-level UFL with penalties}

The approximation algorithm $A$ presented below is parameterized by $\gamma_l$.
\begin{algorithmic}[1]
 \STATE formulate and solve the extended LP (12)-(17) to get an optimal solution  $(x^*, g^*)$;
 \STATE scale up facility opening and client rejecting variables by $\gamma_l$,
 then recompute values of $x^*_{p}$ for $p \in P_C$ to obtain a minimum cost solution $(\bar{x}, \bar{g});$\\
 \STATE divide clients into two groups $C_{\gamma_l} = \{ j \in C | \gamma_l \cdot (1 - g_j^*) \geq 1 \}$ and $\bar{C}_{\gamma_l} = C \setminus C_{\gamma_l};$
 \STATE cluster clients in $C_{\gamma_l}$;
 \STATE round facility opening (tree by tree);
 \STATE connect each client $j$ with a closest open connection path unless rejecting it is a cheaper option.
\end{algorithmic}
Our final algorithm is as follows: run algorithm $A(\gamma_l)$ for each $l = 1, 2 \dots, n-1$ and select a solution with the smallest cost.

Clustering is based on rules described in \cite{Chudak} which is generalized in \cite{Rybicki} for $k$-level instances.
Rounding on a tree was also used in~\cite{Rybicki}. Nevertheless, for completeness we give a brief description of step 4 and 5 in the following subsections.
From now on we are considering only scaled up instance $(\bar{x}, \bar{g})$.

\subsection{Close and distant facilities}

For any client $j \in C_{\gamma}$, let $P^j$ be the set of top-level facilities fractionally serving $j$ in $(\bar{x}, \bar{g})$.
As discussed in Section~\ref{sec:one_level}, WLOG the fractional connectivity of $j$ to a set of facilities may be assumed to be
the fractional opening of these facilities.
Sort facilities $i_1, i_2, \dots i_m$ {from} $P^j$ by non-decreasing distance from client $j \in C_{\gamma}$,
and select the smallest subset of $P^j$ with volume one - this is the set of close facilities $P_c^{j}$,
the rest of facilities {from} $P^j$ are distant facilities $P_d^{j}$.
By $D_{av}^{C}(j), D_{av}^{D}(j)$ and $D_{av}(j)$ we denote the average distances from $j$ to close, distant and all facilities in set $P^j$ {respectively}.
Moreover by $D_{max}^{C}(j)$ we denote the maximal distance from $j$ to a close facility. Formal definitions are as follows:
$$D_{av}^{C}(j) = \frac{\sum_{p \in P_c^{j}} c_p \bar{x}_p}{\sum_{p \in P_c^{j}}\bar{x}_p} = \sum_{p \in P_c^{j}} c_p \bar{x}_p; \; \; \; D_{av}^{D}(j) = \frac{\sum_{p \in P_d^{j}} c_p \bar{x}_p}{\sum_{p \in P_d^{j}}\bar{x}_p} = \frac{\sum_{p \in P_d^{j}} c_p \bar{x}_p}{\gamma (1 - g_j^*) - 1}.$$

Using the similar arguments as in \cite{Ghodsi} we can define $\rho_j = \frac{D_{av}(j) - D_{av}^{C}(j)}{D_{av}(j)}$ and express $D_{av}^{C}(j)$ and $D_{av}^{D}(j)$ using $\rho_j$. $$D_{av}^{C}(j) = (1 - \rho_j) D_{av}(j); \; \; \; D_{av}^{D}(j) = (1 + \frac{\rho_j}{\gamma(1 - g_j^*) - 1}) D_{av}(j).$$

\subsection{Clustering}

Two clients $j_1, j_2 \in C_{\gamma}$ are called neighbors if $P_c^{j_1} \cap P_c^{j_2} \neq \emptyset$.
\begin{algorithmic}[1]
 \WHILE{there is an unclustered client in $C_{\gamma}$}
  \STATE select unclustered client $j \in C_{\gamma}$ that minimizes $D_{av}^{C}(j) + D_{max}^{C}(j)$,
  \STATE form a new cluster containing $j$ and all its unclustered neighbors from $C_{\gamma}$,
  \STATE call $j$ the center of the new cluster;
 \ENDWHILE
\end{algorithmic}

The above clustering procedure (just like in \cite{Chudak}) partitions all clients into groups called clusters. Such partition has two important properties.
First: there are no two neighbors from $C_{\gamma}$ which are (both) centers of clusters.
Second: distance from any client in cluster to his cluster center is not too big.

\subsection{Randomized facility opening}

Consider an arbitrary cluster center $j$. Since LP solutions have a form of a forest, we only need to focus on rounding single tree serving $j$. For clarity, within this rounding procedure we will refer to facilities as vertices (of a tree), and use $x_v$ to denote the fractional opening of vertex (facility) $v$ and $y_v$ to denote the extent in which the cluster center $j$ uses $v$ in $(\bar{x}, \bar{g})$, i.e, $y_v = \sum_{p \in P^{j} : v \in p} \bar{x}_p$. Note that $x_v\geq y_v$ for each $v$ and $x_v\leq x_{father(v)}$ if $v$ is not the root of a tree.

The main idea is to open exactly one path for cluster center $j$ but keep the probability of opening of each vertex $v$ equal to $x_v$ in the randomized procedure.
In \cite{Rybicki} we gave a token-passing-based adaptation of the procedure by Garg Konjevod and Ravi~\cite{Garg}, that stores the output in $\hat{x}$ and $\hat{y}$,
and has exactly the desired properties.

\begin{lemma}
 $E[\hat{x}_v] = x_v$ and $E[\hat{y}_v] = y_v$ for all $v \in V$.
\end{lemma}

It is essential that the probability of opening at least one path in a set $B_j \subseteq \{p \in P_C~|~j \in p\}$ can be lower bounded by a certain function $F_k(x)$, where $x$ is the total flow from client $j$ to all paths in $B_j$ and $k$ is the number of levels in the considered instance. It can be shown that $F_1(x) \geq 1 - e^x$ and the following lemma (from \cite{Rybicki}) hold.
For more details see Appendix \ref{sec:functions of probability} and \cite{Rybicki}.

\begin{lemma}
 Inequality $F_{k}(x) \geq 1 - e^{(c-1)x}$ implies $F_{k+1}(x) \geq 1 - e^{(e^{c-1} - 1)x}$.
\end{lemma}

\section{Analysis}
\va{The high level idea is that we can consider the instance of $k$-level UFLWP as a corresponding instance of $k$-level UFL 
by showing that the worst case approximation ratio is for clients in set $C_{\gamma}$ and we can treat the penalty of client $j\in C_{\gamma}$ 
as a ``penalty-facility" in our analysis. That is, we can overcome penalties by solving an equivalent $k$-level UFL without penalties.}

\subsection{Complete solution and ``one-level'' description} \label{sec:one_level}

It is standard in uncapacitated location problems to split facilities to obtain a so called \emph{complete} solution,
where no facility is used less than it is open by a client (see \cite{Sviridenko} for details). For our algorithm, to keep
the forest structure of the fractional solution, we must slice the whole trees instead of splitting individual facilities to obtain the following.

\begin{lemma} \label{completeness}
 Each solution of our linear program for $k$-level UFLWP can be transformed to an equivalent complete solution.
\end{lemma}
\begin{proof}
We should give two copies $T'$ and $ T''$ of tree $T$ (instead of it) if there is some client $j \in C$ with a positive flow $x_{jp}$ to one of the paths $p$ in the tree $T$ which is smaller than the path opening $x_p$. Let the opening of such problematic path be equal to flow $x_{jp}$ in tree $T'$. In tree $T''$ it has value equal to the opening in $T$ decreased by $x_{jp}$. In general each facility in tree $T'$ ($T''$) has the same opening as in $T$ times $\frac{x_{jp}}{x_p}$ ($\frac{x_p - x_{jp}}{x_p}$). Note that the value of flow from client $j$ (and other clients which are connected with both trees now) should be the same as before adding trees $T'$ and $T''$ instead of $T$. All clients ``recompute" their connection values.
We sort all paths in increasing connection cost for client $j$ and connect with them (in that order) as strong as it is possible until client $j$ has flow equal to one or it is cheaper to pay penalty instead of connecting with any open path. The important fact is that the expected connection and penalty cost of each client remain the same after above operations.

In the process of coping and replacing trees we add at most $|C|$ new trees. Because each client has at most one {``problematic''} (not saturating) path.
\qed
\end{proof}

For the clarity of the following analysis we will use a ``one-level" description of the instance and fractional solution despite its $k$-level structure. Because the number of levels will have influence only on the probabilities of opening particular paths in our algorithm.

Consider set $S_j$ of paths which start in client $j$ and end in the root of a single tree $T$. Instead of thinking about all paths from set $S_j$ separately we can now treat them as one path $p_T$ whose fractional opening is $x_{p_T} = \sum_{p \in S_j} \bar{x}_p$
and (expected) cost is $c_{p_T} = \frac{\sum_{p \in S_j} c_p \bar{x}_p}{x_{p_T}}$. Observe that our distance function $c_{p_T}$ satisfy the triangle inequality. From now on we will think only about clients and facilities (on level $k$) and (unique) paths between them. Accordingly, we will now encode the fractional solution as $(\bar{x}, \bar{y}, \bar{g})$, to denote the fractional connectivity, opening and penalty components.

\subsection{Penalty discussion}
\begin{lemma}
 \label{ineq_proof}
 {$\forall_{\gamma>1, 1 \geq g_j^* \geq 0}~D_{max}^C(j) \leq \gamma (1-g_j^*) D_{av}(j) + (3 - \gamma (1 - g_j^*)) D_{max}^C(j)$.}
\end{lemma}

\begin{proof}
 $$D_{max}^C(j) \leq D_{av}^C(j) + 2 D_{max}^C(j) $$ $$ \leq D_{av}^C(j) + (\gamma (1-g_j^*) - 1) D_{av}^D(j) + (3-\gamma (1-g_j^*)) D_{max}^C(j) $$ $$= \gamma (1-g_j^*) D_{av}(j) + (3 - \gamma (1-g_j^*)) D_{max}^C(j)$$
 \\The second inequality holds because $D_{max}^C(j) \leq D_{av}^D(j)$. Moreover to justify the last equality we should observe that $D_{av}(j) = \frac{1}{\gamma(1-g_j^*)}D_{av}^C(j) + \frac{\gamma(1-g_j^*) - 1}{\gamma(1-g_j^*)} D_{av}^D(j)$.
 \qed
\end{proof}

\begin{lemma}
 \label{worst_case}
 The worst case approximation ratio is for clients from set $C_{\gamma}$.
\end{lemma}

\begin{proof}
 We have two types of clients divided for two sets $C_{\gamma}$ and $\bar{C_{\gamma}}$. Lets sort facilities in nondecreasing distances from client $j$. In that proof $l$ is number of facilities which has positive flow from $j$ in considering (scaled up) fractional solution. Suppose the first case $q \in C_{\gamma}$, then we can upper bound his connection and penalty cost like that 
 $$E[C_q + P_q] \leq \sum_{i = 1}^{l}(F_k(\sum_{j = 1}^{i} y_j) - F_k(\sum_{j = 1}^{i-1} y_j)) d(q,j) $$ $$ + (1 - F_k(\gamma (1 - g))) (\gamma g_j p_j + (max \{0, 1 - \gamma g_j\}) (\gamma (1 - g) D_{av}(q) + (3 - \gamma (1 - g)) D_{max}^C(q)))$$ $$\leq \sum_{i = 1}^{l}(F_k(\sum_{j = 1}^{i} y_j) - F_k(\sum_{j = 1}^{i-1} y_j)) d(q,j) $$ $$ + (1 - F_k(\gamma (1 - g))) (\gamma (1 - g) D_{av}(q) + (3 - \gamma (1 - g)) D_{max}^C(q)) $$ Inequality holds because $p_q \leq \gamma (1 - g) D_{av}(q) + (3 - \gamma (1 - g)) D_{max}^C(q)$ in other case $q$ would be connected with facility in that distance instead of using penalty.
 
 In the second case we have that $q \in \bar{C_{\gamma}}$. Connection and penalty cost of client $q$ can be upper bounded in below way
 $$E[C_q + P_q] \leq (F_k(\sum_{j = 1}^{l} y_j) - F_k(\sum_{j = 1}^{l-1} y_j)) d(q,j) + (1 - F_k(\gamma (1 - g))) p_q$$ Note that for each client $j \in \bar{C_{\gamma}}$ the truth is $p_j \leq D_{max}^C(j)$, so from Lemma \ref{ineq_proof} we have that the worst case approximation ratio is for clients from set $C_{\gamma}$.
 \qed
\end{proof}

\begin{lemma}
 \label{g_0}
 For clients $j \in C_{\gamma}$ we can treat its penalty as a facility.
\end{lemma}

\begin{proof}
If $j$ is a cluster center, $j$ will have at least one (real) facility open in {its} set of close facilities. Thus, its connection and penalty cost are independent of the value of $g_j^*$. If $j$ is not a cluster center and we pretend
its penalty as a facility, no other client $j'$ will consider to use this fake facility. Because $j'$ only looks at facilities
fractionally serving him, and the facilities which serve the center of the cluster containing $j'$.
\qed
\end{proof}

\subsection{Approximation ratio}
\label{group_analize}

A single algorithm $A(\gamma)$ has expected facility opening cost $E[F]\leq \gamma \cdot F^*$ and expected connection and penalty cost $E[C+P] \leq max \{3 - 2 \cdot F_k(\gamma), \frac{2 - F_k(\gamma) - F_k(1)}{1 - \frac{1}{\gamma}}\} \cdot (C^* + P^*)$ (see Appendix \ref{sec:single_algorithm} for a detailed proof). To obtain an improved approximation ratio we run algorithm $A$ for several values of $\gamma$ and select the cheapest solution. The following LP gives an upper bound on the approximation ratio.

\begin{eqnarray}
  \label{k_lp:max}
  max~T && \\
  \label{k_lp:t_upperbound}
  \gamma_i f + \sum_{l = 1}^{n} c_l \cdot p_l^i + (1 - e^{-\gamma_i}) (\gamma_i c + (3 - \gamma_i)c_{i+1}) \geq T &&  \forall_{i < n}\\
  \label{k_lp:c_constraint}
  \frac{1}{\gamma_1} \cdot c_1 + \sum_{i = 2}^{n}(\frac{1}{\gamma_i} - \frac{1}{\gamma_{i-1}}) \cdot c_i = c && \\
  \label{k_lp:c_order}
  0 \leq c_i \leq c_{i+1} \leq 1&&\forall_{i < n}\\
  \label{k_lp:opt_sol}
  f + c = 1 && \\
  f, c\geq 0 &&
\end{eqnarray}

Since the number of levels has influence on connection probabilities, the values of $p_l^i$ need to be defined more carefully than for UFL.
In particular, for $l = 1$ we now have $p_1^i = 1 - F_k(\frac{\gamma_i}{\gamma_1})$ and $p_l^i = F_k(\frac{\gamma_i}{\gamma_{l-1}}) - F_k(\frac{\gamma_i}{\gamma_{l}})$ for $l > 1$.

The Table \ref{improved_ratios} summarizes the obtained ratios for a single algorithm (run with the best choice of $\gamma$ for particular $k$) and for a group of algorithms.

\begin{figure}
  \centering
  \includegraphics[height=57mm]{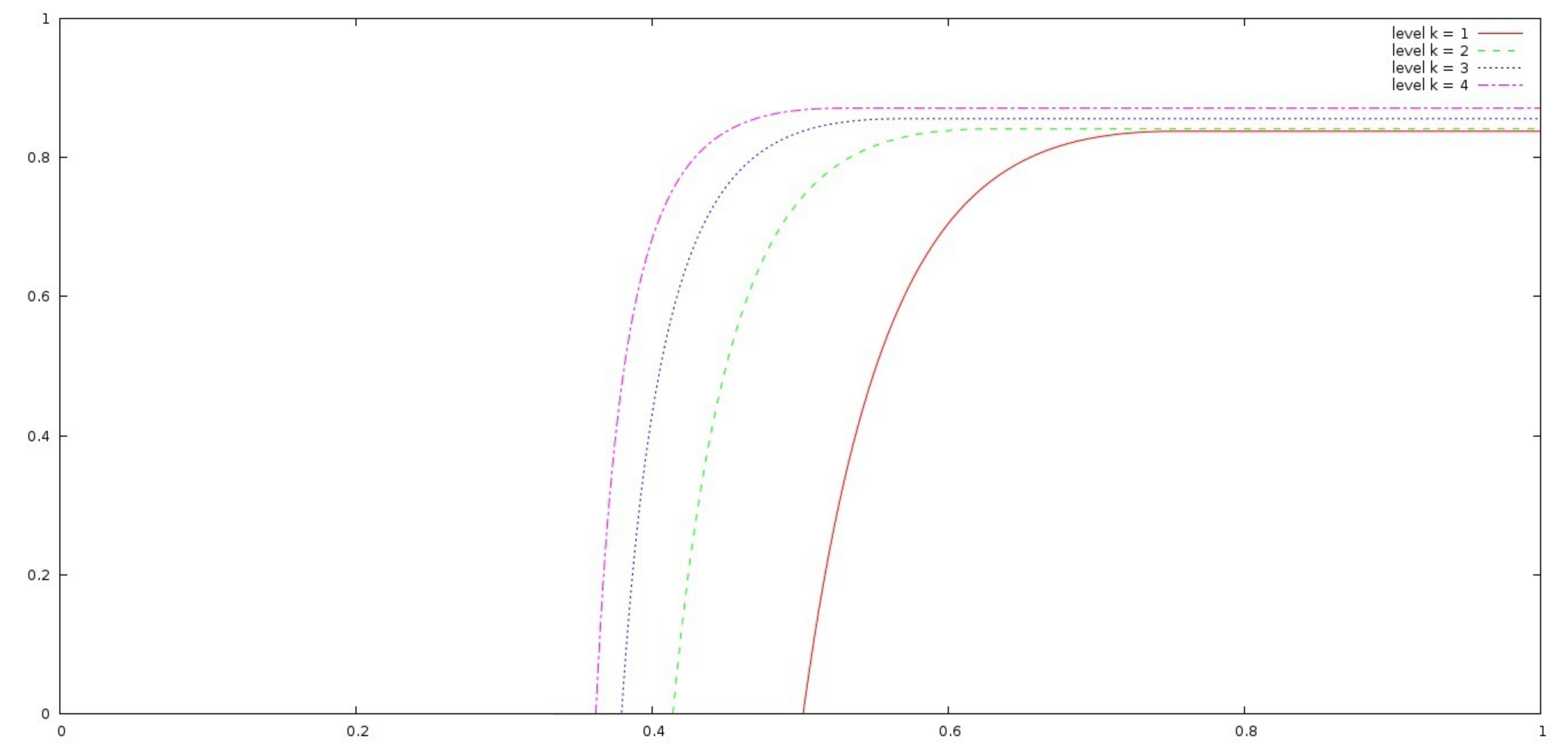}
  \caption{Worst case profiles of $h(p)$ (i.e., distances to facilities) for $k = 1, 2, 3, 4$ obtained from solution of the LP in section \ref{group_analize}. \va{X-axis is volume of a considered set and y-axis represents distance to the farthest facility in that set. Values of function $h(p)$ are in one-to-one correspondence with values of $c_i$ in LP from section \ref{group_analize}.}}
  \label{fig:hardcase}
\end{figure}
\begin{figure}
\begin{minipage}{6cm}
  \centering
  \includegraphics[width=6cm]{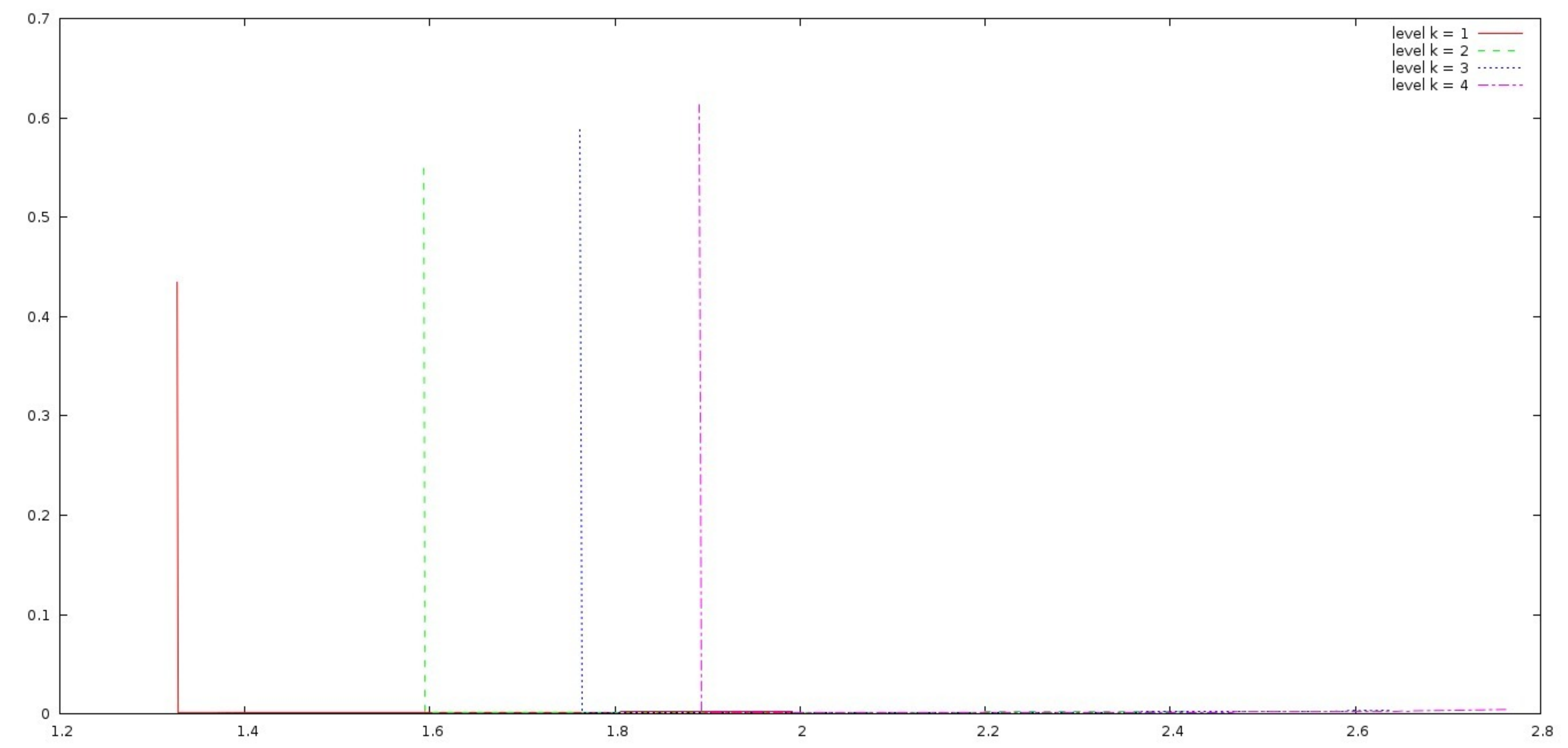}
\end{minipage}
\begin{minipage}{6cm}
  \centering
  \includegraphics[width=6cm]{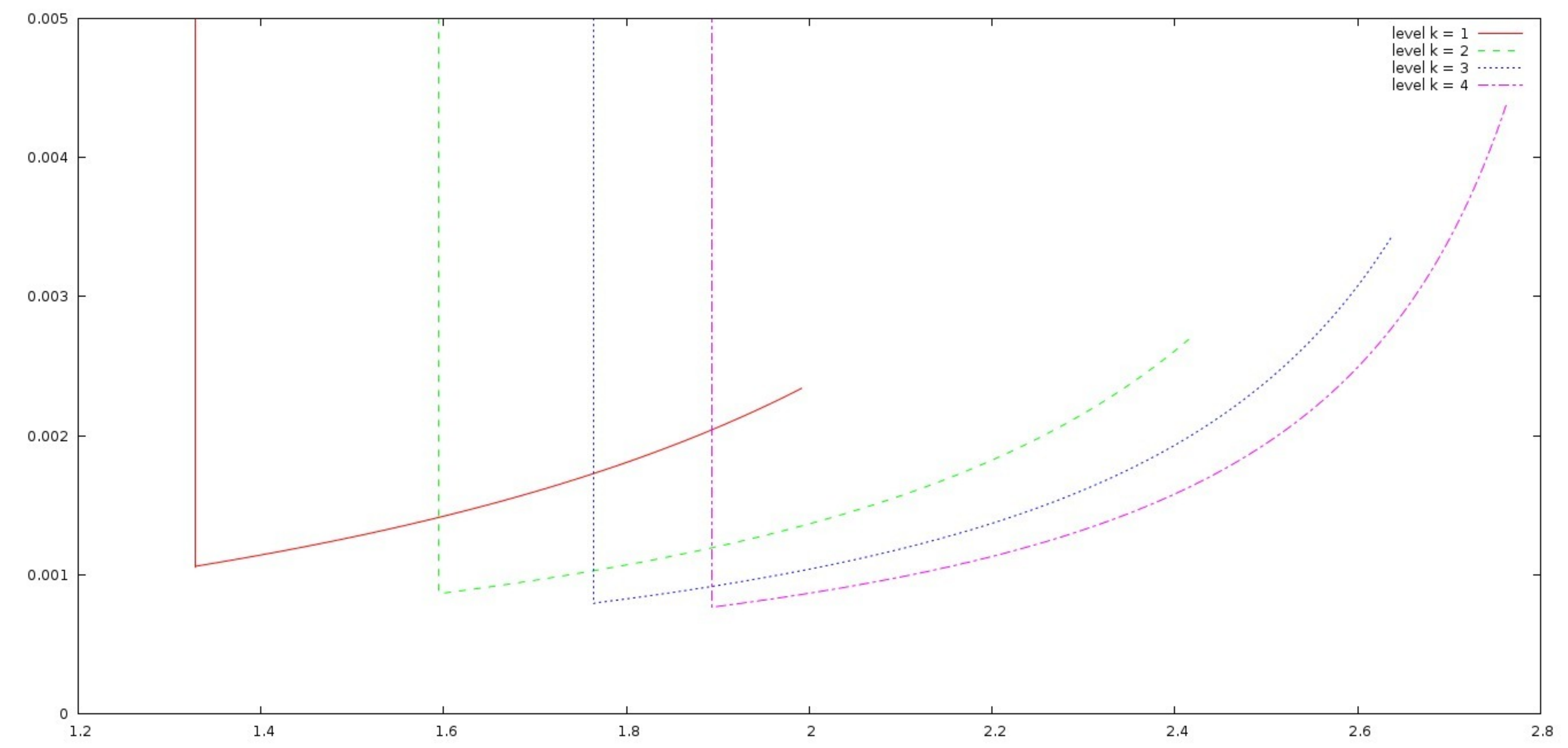}
\end{minipage}
 \caption{Probabilities of using a particular $\gamma$ in a randomized alg. (from the dual LP) for $k = 1, 2, 3, 4$.
 Left figure: general view; Right figure: close-up on small probabilities.}
  \label{fig:probabilities}
\end{figure}

\newpage
{\large \bf APPENDIX}
\begin{appendix}

\section{Graph modification} \label{sec:modiffication}
The idea is to construct a graph in which each facility on level $t$ is connected with exactly one facility on level $t+1$. We will describe in a few words how to do it, but the best idea is to read section 2 in \cite{Rybicki}. Let $F'$ and $F$ be the set of facilities before and after modification respectively. For the highest level nothing change which means $F'_{l_k} = F_{l_k}$. For each facility $i \in F'_{l_{t-1}}$ we have $|F_{l_t}|$ copies each connected with a different facility in $F_{l_t}$. The cardinality of set $F_{l_{t-1}}$ is equal to $|F_{l_t}|\cdot|F'_{l_{t-1}}|$. In general: for each $t = 1, 2, \cdots , k -1$ set $F_{l_t}$ has $|F_{l_{t+1}}|$ copies of each element in set $F^{'}_{l_t}$ and each copy is connected with a different element in the set $F_{l_{t+1}}$. Note that there is an optimal integral solution with the form of a forest. So we do not lose anything important for this optimal solution by modifying the graph in a way described above.

\begin{figure}
  \centering
  \includegraphics[height=35mm]{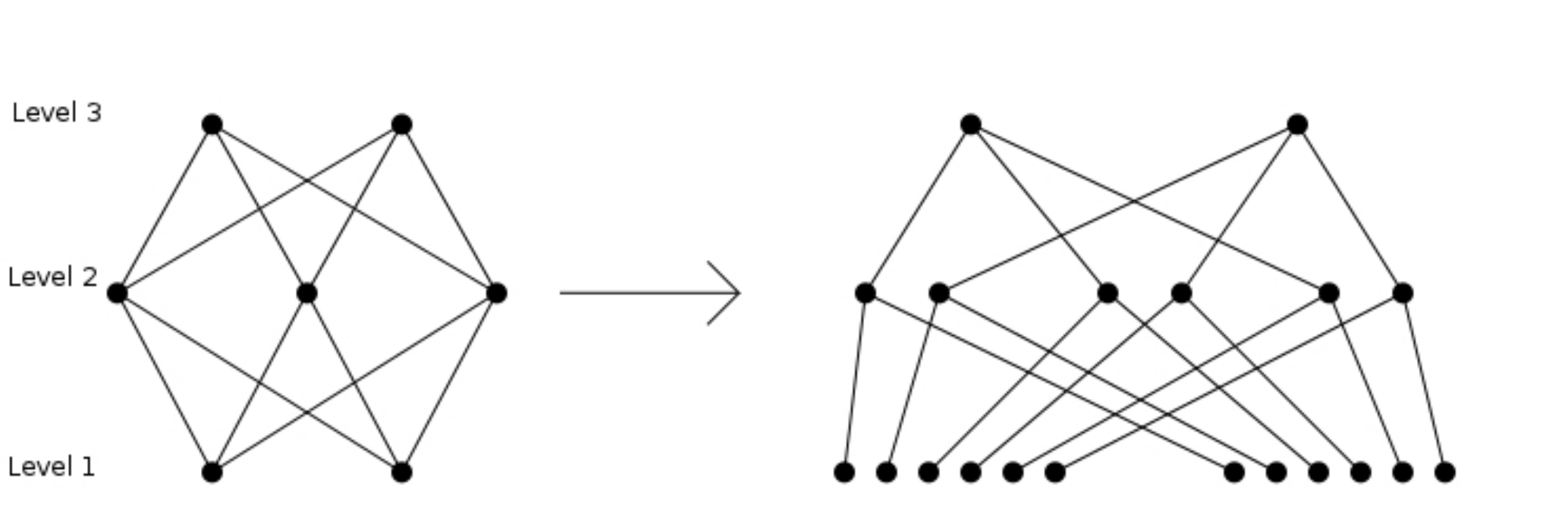}
  \caption{Figure presents graph modification.}
\end{figure}

\section{Functions $f_k(\cdot)$ and $F_k(\cdot)$}\label{sec:functions of probability}
Lets consider set $S_j \subseteq P^j$ of paths which start in client $j \in C$ and end in the root of a tree $T$. We say that client $j$ has flow of value $z$ to tree $T$ if the total value of paths in set $S_j$ is equal to $z$. Byrka et al. in \cite{Rybicki} gives the following definition of function which is a lower bound for the probability of at least one path of a tree will be open as a result of rounding procedure on that tree. We use $max_{x}$ to denote $max_{x_{1}+\ldots+x_{n} = x, x_i > 0}$, similar for $min_{x}$.

\[
f_{k}(z) =
   \begin{dcases*}
     z &  when  k = 1 \\
     z \cdot min_{z} (1 - (\prod_{i = 1}^{n}(1 - f_{k-1}(\frac{z_{i}}{z})))) & other cases

   \end{dcases*}
\]

It is a product of the probability of opening the root node, and the (recursively bounded) probability that at least one of the subtrees has an open path, conditioned on the root being open. Now we are ready to give a function $F_{k}(x)$ to bound the probability of opening at least one path when we have flow $x$ from one client to more than one tree. Let $F_{k}(x) = 1 - max_{x} \prod_{i = 1}^{n}(1 - f_{k}(x_{i}))$, which is one minus the biggest chance that no tree gives a route from the root to a leaf, using the previously defined $f_k(.)$ function to express the success probability on a single tree.

\section{Analysis of single algorithm}
\label{sec:single_algorithm}

Now we can upper bound the expected connection and penalty cost of single algorithm. As it was proved in Lemma \ref{worst_case} the worst case scenario is for client $j \in C_{\gamma}$ which is not a cluster center, so to upper bound the expected connection and penalty cost we can concentrate on clients from $C_{\gamma}$. Moreover from Lemma \ref{g_0} we can suppose that $g^*_j = 0$.

\begin{lemma}
 The expected connection and penalty cost could be bounded in following way $E[C+P] \leq max \{3 - 2 \cdot F_k(\gamma), \frac{2 - F_k(\gamma) - F_k(1)}{1 - \frac{1}{\gamma}}\} \cdot (C^* + P^*)$.
\end{lemma}

\begin{proof}
The value of $p_c = F_k(1)$ is a chance that at least one facility will be open in the set of close facilities. $p_d = F_k(\gamma) - F_k(1)$
expresses the chance that at least one distant facility of the considered client is open, but all close facilities are closed. The remaining $p_s=1-p_c-p_d$ is the probability of connecting the considered client to the open facility by its cluster center. The cost of this connection is bounded in Lemma \ref{cluster_close_distance}. Suppose $j^\prime\in C$ is the cluster center of $j\in C$.
$$E[C_j + P_j] \leq p_c \cdot D_{av}^{C}(j) + p_d \cdot D_{av}^{D}(j) + p_s \cdot (\gamma D_{av}(j) + (3 - \gamma)D_{max}(j)))$$
$$\leq (p_c + p_s) \cdot D_{av}^{C}(j) + (p_d + 2p_s) \cdot D_{av}^{D} $$
$$= (p_c + p_s) \cdot (1 - \rho_j) \cdot D_{av}(j) + (p_d + 2p_s) \cdot (1 + \frac{\rho_j}{\gamma - 1}) \cdot D_{av}(j)$$
$$= max \{1 + 2 \cdot p_s, \frac{1 + p_s - p_c}{1 - \frac{1}{\gamma}}\} \cdot D_{av}(j)$$
$$= max \{3 - 2 \cdot F_k(\gamma), \frac{2 - F_k(\gamma) - F_k(1)}{1 - \frac{1}{\gamma}}\} \cdot D_{av}(j)$$
You can find the justification for above inequalities in \cite{Ghodsi}. Summing over all clients we get the lemma.
\qed
\end{proof}
\end{appendix}
\end{document}